\documentclass[onecolumn,a4paper]{IEEEtran}
\IEEEoverridecommandlockouts
\usepackage[left=1.435cm,right=1.5cm,top=1.85cm,bottom=4cm]{geometry}

\usepackage[utf8]{inputenc}
\usepackage{amsmath,amssymb,amsfonts}
\usepackage{mathtools}
\usepackage{amsthm}
\usepackage{algorithmic}
\usepackage{graphicx}
\usepackage{textcomp}
\usepackage{tikz}
\usepackage{caption}
\usepackage{cuted}
\usepackage{pgfgantt}
\usepackage{pdflscape}
\usepackage{pst-plot}
\usepackage{comment} 
\usepackage{cases}
\usepackage{nccfoots}
\usepackage{lineno,hyperref}
\usetikzlibrary{spy}
\usetikzlibrary{positioning,calc}
\usetikzlibrary{decorations.pathmorphing,calc,shapes,shapes.geometric,patterns}
\usetikzlibrary{shapes.multipart}
\usepackage{xfrac}
\usepackage{colortbl}
\usepackage{cancel} 
\usetikzlibrary{arrows,positioning,calc,intersections}
\usetikzlibrary{datavisualization.formats.functions}
\def\BibTeX{{\rm B\kern-.05em{\sc i\kern-.025em b}\kern-.08em
    T\kern-.1667em\lower.7ex\hbox{E}\kern-.125emX}}
    
\usepackage{romannum}
\usepackage{pgfplots}
\usepgfplotslibrary{fillbetween}
\usetikzlibrary{arrows, decorations.markings}

\newtheorem{theorem}{Theorem}
\newtheorem*{theorem*}{Theorem}
\newtheorem{lemma}[theorem]{Lemma}

\newtheorem{definition}[theorem]{Definition}

\newtheorem{remark}[theorem]{Remark}

\newcommand{\seta}{\ensuremath{\mathcal{A}}}
\newcommand{\setx}{\ensuremath{\mathcal{X}}}
\newcommand{\sety}{\ensuremath{\mathcal{Y}}}
\newcommand{\setz}{\ensuremath{\mathcal{Z}}}
\newcommand{\setd}{\ensuremath{\mathcal{D}}}
\newcommand{\setu}{\ensuremath{\mathcal{U}}}
\newcommand{\sete}{\ensuremath{\mathcal{E}}}

\newcommand{\setp}{\ensuremath{\mathcal{P}}}

\newcommand{\sett}{\ensuremath{\mathcal{T}}}

\newcommand{\var}{\ensuremath{\text{Var}}}

\newcommand{\up}{u^\prime}
\newcommand{\upp}{u^{\prime\prime}}
\newcommand{\Mp}{M^\prime}
\newcommand{\Mpp}{M^{\prime\prime}}
\newcommand{\dc}{\mathcal{D}}
\newcommand{\dcp}{\mathcal{D}^\prime}
\newcommand{\dcpp}{\mathcal{D}^{\prime\prime}}
\newcommand{\C}{\mathcal{C}}
\newcommand{\Cp}{\mathcal{C}^\prime}
\newcommand{\Cpp}{\mathcal{C}^{\prime\prime}}

\newcommand{\mc}{\mathcal}

\newcommand{\circlearrow}{}
\DeclareRobustCommand{\circlearrow}{%
  \mathrel{\vphantom{\rightarrow}\mathpalette\circle@arrow\relax}%
}
\newcommand{\circle@arrow}[2]{%
  \m@th
  \ooalign{%
    \hidewidth$#1\circ\mkern1mu$\hidewidth\cr
    $#1-$\cr}%
}

\DeclarePairedDelimiterX{\infdivx}[2]{(}{)}{%
  #1\;\delimsize\|\;#2%
}

\pgfplotsset{compat=1.18}
\begin{document}
\title{Common Randomness Generation from Sources with \color{black} Infinite Polish \color{black} Alphabet}


\author{
\IEEEauthorblockN{Wafa Labidi \IEEEauthorrefmark{1}\IEEEauthorrefmark{2}, Rami Ezine\IEEEauthorrefmark{1}, Moritz Wiese\IEEEauthorrefmark{1}, Christian Deppe\IEEEauthorrefmark{2} and Holger Boche\IEEEauthorrefmark{1}}
\IEEEauthorblockA{\IEEEauthorrefmark{1}Technical University of Munich\\
\IEEEauthorrefmark{2}Technical University of Braunschweig
}
}

\maketitle

\begin{abstract}
\color{black}We investigate the problem of common randomness (CR) generation in the basic two-party communication setting in which a sender and a receiver aim to agree on a common random variable with high probability. The terminals observe independent and identically distributed (i.i.d.) samples of sources with an arbitrary distribution defined on a \textit{Polish} alphabet and are allowed to communicate as little as possible over a noisy, memoryless channel. We establish single-letter upper and lower bounds on the CR capacity for the specified model. The derived bounds hold with equality except for at most countably many points where discontinuity issues might arise. 

\end{abstract}
 
\section{Introduction}

In the standard two-source model for common randomness (CR) generation, the sender Alice and the receiver Bob aim to generate a common random variable with a high probability of agreement. 
With the availability of CR, it is possible to implement correlated random protocols that often outperform the deterministic ones or the ones using independent randomization \cite{Sudan,part2}.
For instance, CR allows an enormous performance gain in the identification (ID) scheme, which is a novel approach in communications developed by Ahlswede and Dueck \cite{Idchannels} in 1989. In contrast to the classical transmission scheme proposed by Shannon \cite{Shannon}, the resource CR allows a considerable increase in the ID capacity of channels \cite{trafo,part2,ahlswede2021}. \color{black}
The ID approach is much more efficient than the classical transmission scheme for many new applications with high reliability and latency requirements including digital watermarking \cite{MOULINwatermarking,AhlswedeWatermarking,SteinbergWatermarking}, industry 4.0 \cite{industry4.0} and 6G communication systems \cite{6Gcomm}.

Another obvious application of CR generation is in the secret key generation problem since the generated CR is used as a secret key in cryptography \cite{part1}\cite{maurer}. In this paper, however, no secrecy constraints are imposed. It is also worth mentioning that CR is highly relevant in the modular coding scheme for secure communication, where non-secure CR can be used as a seed \cite{semanticsecurity} and \cite{boche2022semantic}.

For all these reasons, CR generation for future communication networks is a central research question in several large 6G research projects \cite{researchgroup1}\cite{researchgroup2}. 
It is expected that CR will play an important role in achieving  the robustness, low-latency, ultra-reliability, resilience  \cite{6Gcomm} \cite{6Gpostshannon} and security requirements \cite{semanticsecurity} imposed by these future communication systems. 

\color{black}Over the past decades, the problem of CR generation from correlated discrete sources with finite alphabet has been investigated in several studies. 
Ahlswede and Csiszár initially introduced this problem in \cite{part2}. \color{black} They considered a two-source model for CR generation, where the sender and receiver communicate over a rate-limited discrete noiseless channel, as well as another two-source model that allows communication over noisy memoryless channels. A single-letter CR capacity formula was derived for both models in \cite{part2}.
Later, \color{black}the results on CR capacity have been extended to Gaussian channels in \cite{globecom}. \color{black}
Recently, the authors in \cite{SISOfasingCR} and in \cite{MIMOfadingCR} focused on 
the problem of CR generation over fading channels. A more general scenario has been investigated in \cite{UCR}, where Alice is allowed to send information to Bob via an arbitrary single-user channel. 
However, to the best of our knowledge, \color{black}no studies apart from the ones in \cite{CRgaussiansources}\cite{CRcountable} \ have addressed the problem of CR generation from sources with infinite alphabet and focused on deriving the CR capacity for such models.

\color{black}In the theory of CR generation, our aim is to construct models that closely resemble reality. This involves examining continuous alphabets, many of which fall under the category of Polish alphabets. \color{black}
The main contribution of our work lies in deriving single-letter lower and upper bounds on the CR capacity of a two-source model involving correlated arbitrary sources $X$ and $Y$ with infinite Polish alphabets \color{black}  encompassing $\mathbb{R}^n$, \color{black} assisted by one-way communication over noisy memoryless channels. This is done under the assumption that the mutual information between $X$ and $Y$ is finite. The transition to infinite Polish alphabets has significant consequences in terms of Shannon entropy convergence, variational distance convergence, etc. \color{black}As we will discuss in the subsequent sections of the paper, we might encounter discontinuity issues when dealing with infinite alphabets.\color{black}

In the proof of our results, we make use of a generalized typicality which is suitable for Polish alphabets. The latter was introduced in \cite{Mitran}. This concept reduces to strong typicality when dealing with finite alphabets. \color{black}Thus, our proof generalizes that of Ahlswede and Csiszár in \cite{part2}.\color{black} 

\quad \textit{ Paper Outline:}  In Section \ref{sec:preli}, we provide the formal definitions that characterize Polish spaces, as well as those characterizing the generalized typicality which was introduced in \cite{Mitran} and which is suitable for Polish alphabets.
 In Section \ref{sec:systemModel}, 
we present the system model for CR generation and provide the key definitions and our main results. \color{black}In Section \ref{proofthmsource1and2}, we prove \color{black} single-letter upper and lower bounds \color{black} on the CR capacity for our specified model. \color{black}
 Section \ref{sec:Conclusions} contains concluding remarks.

\section{Preliminaries} \label{sec:preli}
{\color{black}{\subsection{Polish Alphabets}

\color{black}Polish spaces are fundamental in various areas of mathematics, including topology, analysis, and probability theory.
The term "Polish" pays homage to the pioneering contributions of Polish mathematicians such as Sierpiński, Kuratowski, and Tarski, who extensively studied these spaces.
A Polish space denotes a topological space that is both \emph{completely metrizable} and \emph{separable} \cite{klenke2020probability,dembo2009large}. This means it can be equipped with a metric such that every sequence in the space has a convergent subsequence, and it contains a countable, dense subset. \color{black}
Examples of Polish spaces include the set of real numbers ${\mathbb R}$ equipped with the standard Euclidean metric,  Euclidean spaces $\mathbb{R}^n$ with the standard topology \cite{klenke2020probability}, {\color{black}{finite sets}}, \color{black}the space of continuous functions over compact intervals endowed with the supremum norm, etc. 
Now, we define the weak convergence, as it will be significant for what follows.
We denote the weak convergence of a sequence of measures $P_n$ to a measure $P$ by $ \text{w-lim}_{n\to \infty} P_n = P.$
In certain texts, weak convergence might be represented using a different notation "weak* convergence" as used in \cite{mitran_polish}.
The Portemanteau theorem \cite[Theorem 13.16]{waldequation} provides equivalent \color{black}conditions for weak convergence \color{black} that will be used subsequently.
\begin{theorem}{(Portemanteau)}
    Let $\mathcal{X}$ be a Polish space with Borel $\sigma-$algebra $\sete_\setx$. A bounded sequence of probability measures $P_n, \quad n = 1 , 2 , \ldots$  on $(\setx,\sete_\setx)$ is said to converge weakly to a probability measure $P$ if any of the following equivalent conditions is true.
\begin{enumerate}
    \item $\text{w-lim}_{n\to \infty} P_n = P.$
    \item $\lim_{n\to \infty} \int f\ dP_n = \lim_{n\to \infty} \int f\ dP $  for all bounded continuous functions $f$.
    \item {\color{black}{$\lim_{n\to \infty} \int f\ dP_n = \lim_{n\to \infty} \int f\ dP $}} for all bounded measurable functions $f$ with $P(\setd_f)=0$, where $\setd_f$ is the set of points of discontinuity of $f$.
    
    \item $\lim_{n\to \infty} P_n(\seta)=P(\seta)$, for all $\seta \subset \sete_\setx$ with $P(\partial \seta)=0$, where $\partial\seta$ denotes the topological boundary of $\seta$.
\end{enumerate} 
\end{theorem}
It is worth mentioning that the second condition is commonly regarded as the definition of weak convergence. Weak convergence, introduced above, \color{black}induces weak topology on the set \color{black} of probability measures in $(\setx,\sete_\setx)$ \cite{klenke2020probability}.

 \color{black}
}}

\subsection{Typicality Criteria for Polish Alphabets} \label{subsec: typicality}
In this section, we recall the definition of generalized typicality (w.r.t. weak convergence) on Polish spaces introduced in \cite{Mitran}.
\begin{definition}
Given a sequence $x^n=(x_1,\ldots,x_n) \in \setx^n$, we define the associated empirical distribution $P_{x^n}$ as 
$P_{x^n}(A) := \frac {1}{n} \sum _{\ell =1}^{n} {\boldsymbol 1}_ {\{x_\ell \in A\}}.$
Analogously, given two sequences $x^n=(x_1,\ldots,x_n) \in \setx^n$ and $y^n=(y_1,\ldots,y_n) \in \sety^n$, the joint empirical distribution $P_{x^ny^n}$ is defined by
$P_{x^ny^n}(A \times B) := \frac {1}{n} \sum _{\ell =1}^{n} {\boldsymbol 1}_ {\{x_\ell \in A\}} {\boldsymbol 1}_ {\{y_\ell \in B\}}.$  
\end{definition}
A sequence $x^n=(x_1,\ldots,x_n)$ is called typical if its empirical distribution is close to some probability measure. In \cite{mitran_polish} closeness is assessed w.r.t. the weak topology. In particular, let $d_{\mathcal{P}(\setx)}(\cdot,\cdot)$ denote any metric on the space of probability measures $\mathcal{P}(\setx)$ that generates the weak topology. An example of such a metric is the Prohorov metric \cite{waldequation} defined below.
\begin{definition}
Let $\mu, \nu$ be two \color{black}probability \color{black} measures on a Polish space $\setx$ that has metric $d_\setx(\cdot,\cdot)$. \color{black} The Prohorov metric \cite{mitran_polish,klenke2020probability,dembo2009large}
 is defined as \color{black} $d^{\sf Pr}_ {\setx} (\mu , \nu ) := \max \{d'(\nu , \mu ), d'(\mu , \nu )\},$
\begin{equation*} \text{where }d'(\mu , \nu ) := \inf \{\epsilon > 0 \;|\; \mu (A) \leq \nu (A^{\epsilon }) + \epsilon ,\forall A \in {\cal E} \},\qquad \end{equation*} and
$A^{\epsilon } := \{x \in \setx \;|\; d_{\setx}(x, x^\prime) < \epsilon ~\text {for some } x^\prime \in A \}.$ 
\end{definition}
\color{black}In the following, we revisit the definition of typical sequences w.r.t. the generalized typicality notion introduced by Mitran in \cite{Mitran}.\color{black}
\begin{definition}
Let $d_{\setp(\setx)}(\cdot,\cdot)$ be any metric on the space of probability measures $\mathcal{P}(\setx)$ that induces the weak topology, e.g., the Prohorov metric. A sequence $x^n=(x_1,\ldots,x_n)$ is said to be $(P_X,\epsilon)-$typical w.r.t. the generalized typicality described above if$ d_\setx(P_{x^n},P_X)  < \epsilon. $
We represent the set of such typical sequences of length $n$ by $\sett_{\epsilon}^n(P_X)$.
\end{definition}
 If the set $\setx$ is finite, the definition of the aforementioned generalized typicality coincides with the definition of strong typicality \cite{mitran_polish} (except for the occasional requirement that \( P_X(a) = 0 \) if \( P_X(a) = 0 \)). Thus, typical sequences w.r.t. the generalized typicality introduced by Mitran can be considered as a generalized form of strongly typical sequences. \color{black}The generalized typicality \color{black} in \cite{Mitran} satisfies all the properties below, including Markov lemma. This typicality is sufficiently general to be applied to any distribution (discrete, continuous, or mixed), provided that the alphabet is a Polish space. 

\color{black} Theorem \ref{theorem:typicality}, Theorem \ref{theorem:mitranJoint} and Theorem \ref{theorem:mitranLargedev} outline the desirable properties that typical
and jointly typical sequences, w.r.t. Mitran typicality notion \cite{mitran_polish}, should possess. \color{black}
\begin{theorem}[\cite{Mitran}] The following three statements hold. \label{theorem:typicality}
\begin{enumerate}
   \item \color{black} Let $X_1,X_2,\ldots,X_n$ be independent RVs with values in $\setx$ and identical distribution $P_X$. Then for any $\epsilon>0,$
    \begin{equation} \lim _{n \rightarrow \infty } P\left ({ {{X}} ^{n} \in \sett_\epsilon ^{n}(P_{X}) }\right ) = 1. \end{equation} 

\item \color{black} For any $\epsilon > 0$, there is an $\bar{\epsilon}(\epsilon, P_{XY}) > 0$ such that if $(x^n,y^n) \in \sett^n_{\bar{\epsilon}(\epsilon, P_{XY})}(P_{XY})$ then $x^n \in \sett^n_{\epsilon}(P_X)$.
\end{enumerate}    
\end{theorem}
\color{black}
\begin{theorem}[\cite{Mitran}] \label{theorem:mitranJoint}
    Let $x^n$ be an input sequence to a stationary memoryless channel $P_{Y|X}$ \color{black}such that $x\mapsto P_{Y|X}(\cdot|x)$ is continuous \color{black} and let $Y^n$ be the corresponding output sequence. For every $\epsilon > 0$ and $\delta > 0$, there exists an $\overline{\epsilon}(\epsilon,\delta,P_{X},P_{Y|X})$ such that if $x^n \in \sett^n_{\overline{\epsilon}(\epsilon,\delta,P_X,P_{Y|X})}(P_X)$ for all $n$ greater than some $N$, then
    \begin{equation} \liminf _{n \rightarrow \infty } \Pr\left \{{ ( {{x}} ^{n}, {{Y}} ^{n}) \in \sett^{n}_\epsilon (P_{XY}\}}\right ) > 1 - \delta . \end{equation} \label{theorem:conditionalTyp}
\end{theorem}
\begin{theorem}[\cite{Mitran}] \label{theorem:mitranLargedev}
    Let $P_{X,Y}$ be a joint distribution on $\setx \times \sety$ and $P_X$ and $P_Y$ denote its marginals. Let $y^n$ be a sequence and $X^n$ a random sequence drawn i.i.d. according to $P^n_X$.
    \begin{enumerate}
        \item If $\color{black}I(X;Y)\color{black}$ is finite, then for each $\delta > 0$, there are $\epsilon(\delta)$ and $\bar{\epsilon}(\delta)$ such that if $\epsilon < \epsilon(\delta)$, $\bar{\epsilon} < \bar{\epsilon}(\delta)$, and $y_n \in \sett^n_{\bar{\epsilon}}(P_Y)$ for all $n$ greater than some $N$, then
        \begin{align} &\limsup _{n} \frac {1}{n} \log \Pr \{( {{X}} ^{n}, {{y}} ^{n}) \in \sett^{n}_{\epsilon }(P_{XY}) \} \nonumber \\ & \quad \leq -I(X;Y) + \delta .  \label{eq:LD1}\end{align} 
        
        \item For each $\delta > 0$ and $\epsilon > 0$, there is an $\bar{\epsilon}(\epsilon, \delta) > 0$ such that if $y_n \in \sett^n_{\bar{\epsilon}(\epsilon, \delta}(P_Y)$ for all $n$ greater than some $N$, then
        \begin{align*} & \liminf _{n} \frac {1}{n} \log \Pr\{ ( {{X}} ^{n}, {{y}} ^{n}) \in \sett^{n}_{\epsilon }(P_{XY}) \} \nonumber\\ & \quad \geq -I(X;Y) - \delta . \end{align*} 
    \end{enumerate} \label{theorem:LargeDev}
\end{theorem}
\color{black} For the proof of Theorem \ref{theorem:typicality}, Theorem \ref{theorem:mitranJoint} and Theorem \ref{theorem:mitranLargedev}, we refer the reader to \cite{Mitran}. \color{black}
\section{System Models, Definitions and Main Results}
\label{sec:systemModel}
In this section, we introduce our system model and present the main results of the paper.

\subsection{System Model and Definitions}
\label{systemmodel}
Let a memoryless bivarite source model with arbitrary distribution $P_{XY}$ and finite mutual information $I(X;Y)$, with generic variables $X$ and $Y$ on $\setx$ and $\sety$, respectively, be given. \color{black}Assume that $x\mapsto P_{Y|X}(\cdot|x)$ is continuous. 
 \color{black}Assume also that $\setx$ and $\sety$ are both Polish spaces. \color{black} The outputs of $X$ are observed only by Terminal $A$ and those of $Y$ only by Terminal $B$. Assume also that the joint distribution of $(X,Y)$ is known to both terminals. Terminal $A$ can send information to Terminal $B$ over \color{black} an arbitrary (not necessarily discrete) noisy memoryless channel $W$. \color{black}The capacity of the channel $W$ is finite and denoted by $C(W)$. There are no other resources available to any of the terminals, as depicted in Fig. \ref{fig:System}. 
\begin{figure}[t]
\centering
\tikzstyle{block} = [draw, rectangle, rounded corners,
minimum height=2em, minimum width=2cm]
\tikzstyle{blockchannel} = [draw, top color=white, bottom color=white!80!gray, rectangle, rounded corners,
minimum height=1cm, minimum width=.3cm]
\tikzstyle{input} = [coordinate]
\usetikzlibrary{arrows}
\scalebox{.85}{
\begin{tikzpicture}[scale= 1,font=\footnotesize]
\node[blockchannel] (source) {$P_{XY}$};
\node[blockchannel, below=2.8cm of source](channel) { memoryless noisy channel};
\node[block, below left=3cm of source] (x) {Terminal $A$};
\node[block, below right=3cm of source] (y) {Terminal $B$};
\node[above=1cm of x] (k) {$K=\Phi(X^n)$};
\node[above=1cm of y] (l) {$L=\Psi(Y^n,Z^n)$};

\draw[->,thick] (source) -- node[above] {$X^n$} (x);
\draw[->, thick] (source) -- node[above] {$Y^n$} (y);
\draw [->, thick] (x) |- node[below right] {$T^n=\Lambda(X^n)$} (channel);
\draw[<-, thick] (y) |- node[below left] {$Z^n$} (channel);
\draw[->] (x) -- (k);
\draw[->] (y) -- (l);

\end{tikzpicture}}
\caption{Memoryless source model with one-way communication over a memoryless channel.}
\label{fig:System}
\end{figure}
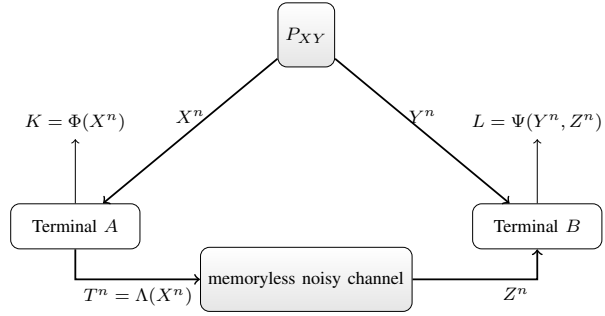
A CR-generation protocol \cite{part2} of block length $n$ and for the above-described system model consists of:
\begin{itemize}
    \item a Lebesgue-measurable function $\Phi$ at Terminal $A$ that maps $X^n$ into a random variable $K$ with \color{black} finite alphabet $\mathcal{K}.$ \color{black}
    \item a Lebesgue-measurable function $\Lambda$ at Terminal $A$ that maps $X^n$ into the \color{black}channel\color{black} \ input sequence $T^n$,
    \item a Lebesgue-measurable function $\Psi$ at Terminal $B$  that maps $Y^n$ and the \color{black}channel\color{black} \ output sequence $Z^n$ into a random variable $L$ with alphabet $\mathcal{K}.$ 
\end{itemize}
This protocol generates a pair of random variables \color{black} $(K,L)$ that is called permissible \cite{part2}.

\begin{definition}  A number $H$ is called an achievable CR rate for the system model in Fig. \ref{fig:System} if there exists a non-negative constant $c$ such that for every $\epsilon>0$ and $\gamma>0$ and for sufficiently large $n$ there exists a permissible  pair of random variables $(K,L)$ such that
\begin{align}
    \Pr\{K\neq L\}\leq \epsilon, 
    \label{errorcorrelated}
\end{align}
\begin{align}
    |\mathcal{K}|\leq 2^{cn},
    \label{cardinalitycorrelated}
\end{align}
\begin{align}
    \frac{1}{n}H(K)> H-\gamma.
     \label{ratecorrelated}
\end{align}
\end{definition}
\begin{definition} 
The CR capacity $C_{CR}(P_{XY},W)$ is the maximum achievable CR rate.
\end{definition}
\subsection{Main Results}
In this section, we characterize \color{black}$C_{CR}(P_{XY},W).$\color{black}
\color{black}
\begin{theorem} \label{theorem:main}
Consider the system model in Fig. \ref{fig:System} with the bivariate source model with arbitrary distribution described in Section \ref{systemmodel}.
 Let 
\color{black} \begin{align}
    &\mc U =
\{P_{U|X}: U \  \text{is defined on a Polish space} \ \text{and} \nonumber \\
    &\quad\quad \quad x\mapsto P_{U|X}(\mc A|x) \ \text{is measurable for every} \nonumber \\
    &\quad \quad \quad \text{measurable} \mc A\subseteq \mc U  \} \label{setUpolish}
\end{align}\color{black}
and
$$L^{(X,Y)}:t \mapsto 	\underset{ \substack{P_{U|X} \in \mathcal{U} \\{\substack{U \circlearrow{X} \circlearrow{Y}\\ I(U;X)-I(U;Y) \leq t}}}}{\sup} I(U;X).$$
Then, the CR capacity $C_{CR}(P_{XY},W)$ satisfies
\begin{equation}
    C_{CR}(P_{XY},W)\geq \underset{\alpha>0}{\sup}\ L^{(X,Y)}\left(C(W)-\alpha\right),
    \label{lowerboundCRcapacity}
\end{equation}
and
\begin{equation}
    C_{CR}(P_{XY},W)\leq \underset{\alpha>0}{\inf} \ L^{(X,Y)}\left(C(W)+\alpha\right),
    \label{upperboundCRcapacity}
\end{equation}
 where $C(W)$ denotes the Shannon capacity of the memoryless channel $W$. The lower and upper bound in \eqref{lowerboundCRcapacity} and \eqref{upperboundCRcapacity} are equal except possibly at the points of discontinuity of the CR capacity of which there are, at most, countably many.
\label{Theorem:Source1}
\end{theorem}
\color{black}
\color{black}
\begin{remark}
    For the case when the joint probability distribution $P_{XY}$ is discrete and $\setx, \ \sety$ are finite, by continuity, the bounds in \eqref{lowerboundCRcapacity} and \eqref{upperboundCRcapacity} coincide as studied in \cite{part2}.
    \end{remark}
    \begin{remark}      
In recent times, research in information theory has made significant progress, proving capacity results for increasingly complex communication scenarios involving arbitrary distributions defined on infinite alphabets. Often, a discontinuity behavior similar to the one in Theorem \ref{theorem:main} occurs  (see \cite{CRMIMOslowfading}, \cite{epsilonUCRcapacity}). For some special cases, interesting solutions are achieved. For instance, unlike in the multiple-input multiple-output (MIMO) case, no discontinuities arise in the characterization of the outage CR capacity of single-input multiple-output (SIMO) slow fading channels with arbitrary state distribution \cite{CRMIMOslowfading}. However, it remains unclear when the discontinuities occur and what values the capacities take at these points. Additionally, determining if there is a general pattern and understanding the reasons behind the discontinuities are still open problems.
    \end{remark}
    \color{black}

\section{Proof of Theorem \ref{Theorem:Source1}}
\label{proofthmsource1and2}
In this section, we provide a proof of Theorem \ref{theorem:main}. For the sake of simplicity, we restrict the proof to the case when the probability distribution $P_{XY}$ is absolutely continuous. This can be analogously extended to arbitrary distributions on Polish spaces.
\subsection{Proof of the Lower-Bound in  Theorem \ref{Theorem:Source1}}
For the achievability part, we extend the coding scheme introduced in \cite{part2} to bivariate source models with arbitrary distribution $P_{XY}$ (here absolutely continuous) and finite mutual information as described in Section \ref{sec:systemModel} \color{black}using generalized typicality \color{black} introduced in Section \ref{subsec: typicality}. We consider the same code construction as used in \cite{part2} based on the same type of binning as for the Wyner-Ziv problem. As we assumed that $I(X;Y)$ is finite, we can easily verify that the quantities $I(U;X),\ I(U;Y)$ are also finite. Thus, Theorem \ref{theorem:LargeDev} can be applied subsequently. 
\color{black}Let $\alpha>0$ be fixed arbitrarily. Let $U$ be a RV  such that $ P_{U|X}\in \mathcal{U}$ and such that $U \circlearrow{X} \circlearrow{Y}$ and $I(U;X)-I(U;Y) \leq C^\prime$ with $C^\prime=C-\alpha.$  \color{black}

{\bf{Code Construction}}: We generate $N_1N_2$ codewords $u^n(i,j),\quad i=1,\ldots,N_1,\ j=1,\ldots,N_2$ by choosing the $n.(N_1N_2)$ symbols $u_l(i,j)$, $l=1,\ldots,n$, independently at random using $P_U$. Each realization $u^n_{i,j}$ of $U^n_{i,j}$ is known to both terminals. For some $\sigma>0$, let $0<\delta_1<2\sigma$. Let $\delta,\bar{\delta},\delta_2,\delta_3 >0$. Let
$N_{1}=2^{\left(n[I(U;X)-I(U;Y)+4\sigma]\right)}$and $
N_{2}=2^{\left(n[I(U;Y)-2\sigma]\right)}.$ 

{\bf{Encoder}}: Let $(x^n,y^n)$ be a realization of $(X^n,Y^n)$. Given $x^n$ with $(x^n,y^n) \in \sett_\delta^n(P_{XY})$, the sender tries to find a pair $(i,j)$ such that $\left(x^n,u^n(i,j) \right) \in \sett_{\delta_1}^n(P_{UX})$. If successful, let $f(x^n)=i$. If no such ${u}^n(i,j)$ exists, then $f({x}^n)=N_1+1$ and $\Phi({x}^n)$ is set to a constant sequence ${u}^n_0$ different from all the ${u}^n(i,j)$s and known to both terminals. We choose $\sigma$ to be sufficiently small such that $\frac{\log(N_1+1)}{n} \leq C^\prime-\sigma',$
for some $\sigma'>0$ and where $N_1+1$ is the cardinality of the set of messages $\{{i}^{\star}=f({x}^n)\}$. The message $i^{\star}=f({x}^n)$, with $i^{\star}\in\{1,\hdots,N_1+1\}$, is encoded to a sequence ${t}^n$ using a suitable code sequence with rate $\frac{\log  (N_1+1) }{n}\leq C^\prime-\sigma'$  and with error probability not exceeding $\frac{\epsilon}{2}$ for sufficiently large $n$. The sequence ${t}^n$ is sent over the channel $W$. 

{\bf Decoder}: Let ${z}^n$ be the channel output sequence. Terminal $B$ decodes the message $\hat{\imath}^{\star}$ from the knowledge of ${z}^n$. Given $\hat{\imath}^{\star}$ and $y^n$, the decoder tries to find ${j}$ such that $\left(y^n,u^n(\hat{\imath}^{\star},{j})\right) \in \sett_{\delta_2}^n(P_{UY})$. If successful, let $L({y}^n,\hat{\imath}^{\star})={u}^n(\hat{\imath}^{\star},{j})$. If there is no such ${u}^n(\hat{\imath}^{\star},{j})$ or there are several, $L$ is set to ${u}^n_0$ (since $K$ and $L$ must have the same alphabet).

{\bf Validity of the requirements in \eqref{errorcorrelated}, \eqref{cardinalitycorrelated} and \eqref{ratecorrelated}:} Let $I^\star=f(X^n)$ be the RV modeling the message encoded by Terminal $A$ and let $\hat{I}^\star$ be the RV modeling the message decoded by Terminal $B$. We first prove that \eqref{errorcorrelated} is satisfied. We have:

    $\Pr\{K\neq L\} 
        \leq \Pr\{K\neq L|I^\star=\hat{I}^\star\}+ \Pr\{I^\star\neq\hat{I}^\star\}.$

We establish an upper-bound on $\Pr\{K\neq L|I^\star=\hat{I}^\star\}$. For this, we consider the probability events $\sete_1, \ \sete_2, \ \sete_3$ and $\sete_4$. 
\begin{enumerate}
\item The source sequences are not jointly typical:
    $\sete_1:=\left\{ (X^n, Y^n) \notin \sett_\delta^n(P_{XY}) \right\}.$

\item For a fixed $x^n$, the encoder cannot find a pair $(i,j)$ such that $(U^n(i,j),x^n) \in \sett_{\delta_1}^n(P_{UX})$:
   $\color{black}\sete_2:= \bigcap_{\substack{i=1,\ldots,N_1 \\j=1,\ldots,N_2}}  \left \{(U^n(i,j), x^n) \notin \sett_{\delta_1}^n(P_{UX}) \right\}.\color{black}$
\item The decoder finds $\tilde{j}\neq j$ such that $\left(y^n,U^n({i},\tilde{j})\right) \in \sett_{\delta_2}^n(P_{UY}) $:
     $\sete_3:=\bigcup_{\substack{\tilde{j}=1,\ldots,N_2 \\ \tilde{j}\neq j}} \left\{ \left(y^n,U^n({i},\tilde{j})\right) \in \sett_{\delta_2}^n(P_{UY}) \right\}.$
 \item 
 The decoder cannot find ${\tilde{j}}$ such that $\left(y^n,U^n({i},{\tilde{j}})\right) \in \sett_{\delta_2}^n(P_{UY}) $:
   $\color{black}\sete_4:= \bigcap_{\tilde{j}=1,\ldots,N_2} \left \{ \left(U^n(i,\tilde{j}),X^n,Y^n\right) \notin \sett_{\delta_3}^n(P_{UXY}) \right\}.\color{black}$
\end{enumerate}
By the union bound, it holds that $\Pr[K\neq L|I^\star=\hat{I}^\star]$ is upper-bounded by
${\color{black}{\Pr[\sete_1]+\Pr[\sete_2]+\Pr[\sete_3]+\Pr[\sete_4 \cap \sete_2^c]}}.$
In the following, we compute an upper-bound on the \color{black} probability of the events {\color{black}{$\sete_1$, $\sete_2$, $\sete_3$ and $\sete_4 $.}}  \color{black}
We have
\begin{align*}
   \Pr\{\sete_1\} &= \Pr \left \{ (X^n, Y^n) \notin \sett_\delta^n(P_{XY}) \right \}\\
  & =1-P_{XY}^n\left(\sett_\delta^n(P_{XY})\right) \overset{(a)}{\leq}\alpha_1(n),
\end{align*}
where $\alpha_1(n)\leq \frac{\epsilon}{6}$ for sufficiently large $n$ and $(a)$ \color{black}follows from Theorem \ref{theorem:typicality}. \color{black}
We have
 \begin{align*}
&\Pr\{\sete_2\} \nonumber \\
&=\Pr\{X^n \notin \sett_{\bar{\delta}}^n(P_X) \}  + \int_{x^n  \in \sett_{\bar{\delta}}^n(P_X)} P_{X^n}(x^n)\nonumber \\
&\quad  \times \Pr\{\bigcap_{\substack{i=1,\ldots,N_1 \\j=1,\ldots,N_2}}  (U^n(i,j), x^n) \notin \sett_{\delta_1}^n(P_{UX})|X^n=x^n\} dx^n\\
&\overset{(a)}{\leq} \kappa(n) \nonumber \\
& \quad + \int_{x^n \in \sett_{\bar{\delta}}^n(P_{X})} P_{X^n}(x^n) \left( 1-2^{-n(I(U;X)+\delta_1) } \right)^{N_1N_2} dx^n \\
    & \overset{(b)}{\leq} \kappa(n)+ \exp(-2^{n(2\sigma-\delta_1)}) \quad \leq \alpha_2(n),
\end{align*}
where $\alpha_2(n) \leq \frac{\epsilon}{6}$ for sufficiently large $n$ and $(a)$  follows from Theorem \ref{theorem:typicality},  because the $U^n(i,j)s$ are independent of $X^n,$ because the $N_1N_2$ events of the intersection are independent and from Theorem \ref{theorem:LargeDev}, $(b)$ follows because $(1-x)^m\leq \exp(-mx)$ and from from the definition of $N_1,N_2.$ 
{\color{black}{We have}}
 \begin{align*}
    \Pr\{\sete_3\}& \overset{(a)}{\leq} \sum_{\tilde{j}\neq j} \Pr\left\{\left(Y^n,U^n({i},\tilde{j})\right) \in \sett_{\delta_2}^n(P_{UY}) \right \} \\
    & \overset{(b)}{<} N_2 \cdot 2^{-n(I(U,Y)+\delta_2)}  = 2^{-n (\delta_2+2\sigma)},
\end{align*}
where $(a)$ follows from the union bound and $(b)$ follows from Theorem \ref{theorem:LargeDev}.  
It holds that
    $\Pr\{\sete_4 \cap \sete_2^c\}\leq \alpha_3(n),$ 
where $\alpha_3(n) \leq \frac{\epsilon}{6}$ for sufficiently large $n$ and where we used the result of Theorem \ref{theorem:conditionalTyp}.
 Now, we know that for sufficiently large $n$, we have  $\alpha_1(n)+\alpha_2(n)+ 2^{-n (\delta_2+2\sigma)}+ {\alpha_3}(n)  \leq \frac{\epsilon}{2}.$
 It follows from the last inequality and the definition of the code sequence $t^n$ that for sufficiently large $n,$ we have

$\Pr\{K\neq L\} 
    \leq \Pr\{K\neq L|I^\star=\hat{I}^\star,\sete^c\}+ \Pr\{I^\star\neq\hat{I}^\star] 
   {\leq}  \epsilon$.
Now, we are going to show that  $(K,L)$ {\color{black}{satisfies \eqref{cardinalitycorrelated} and \eqref{ratecorrelated}.}}
Clearly, \eqref{cardinalitycorrelated} is satisfied  for $c=2\left[I(U;X)+\sigma\right]$, $n$ sufficiently large, since we have$\lvert \mc K \rvert=N_1N_2+1\leq  2^{(2n\left[I(U;X)+\sigma\right])}.$

For any realization $u^n(i,j)$ of $U^n(i,j)$, it holds that
\begin{align*}
    &\Pr\{K={u}^n(i,j)\} \\
&\overset{(a)}{=}\int_{{x}^n\in \sett_{\bar{\delta}}^n(P_X)} \Pr\{K={u}^n(i,j)|X^n={x}^n\}P_{X}^n({x}^n) dx^n  \\
 &\overset{(b)}{\leq} 2^{\left(-n(I(U;X)+\bar{\delta})\right)}. 
\end{align*}
$(a)$ follows because for $(x^n,{u}^n(i,j))$ not jointly typical, we have $\Pr \{K={u}^n(i,j)|X^n={x}^n\}=0$ and $(b)$ follows from Theorem \ref{theorem:LargeDev}. This yields
$H(K)  \geq 
n (I(U;X)+ \bar{\delta}) \geq n H.$ 
This completes the proof of the lower-bound in Theorem \ref{Theorem:Source1}. \qed
\subsection{Proof of the Upper-Bound in Theorem \ref{Theorem:Source1}}
\label{conversethmsource1}
Let $H$ be any achievable CR rate for the system model presented in Section \ref{systemmodel}. So, there exists a non-negative constant $c$ such that for every $\epsilon>0$ and $\gamma>0$ and for sufficiently large $n$ there exists a permissible  pair of RVs $(K,L)$  w.r.t. a fixed CR-generation protocol of block-length $n,$ such that $\Pr[K\neq L]\leq \epsilon,$ $|\mathcal{K}|\leq 2^{cn}$ and $ \frac{1}{n}H(K)> H-\gamma.$
\color{black}
Let $J$ be a random variable uniformly distributed on $\{1,\dots, n\}$ and independent of $K$, $X^n$ and $Y^n$. We further define $U=(K,X_{1},\dots, X_{J-1},Y_{J+1},\dots, Y_{n},J).$ It holds that $U \circlearrow{X_J} \circlearrow{Y_J}$ and that $P_{U|X_J} \in \mc U,$ where $\mc U$ is defined \eqref{setUpolish}. As shown in \cite{part2}, it holds that $\frac{H(K)}{n}\leq I(U;X_J)$ and that $\frac{H(K)}{n}=I(U;X_J)-I(U;Y_J).$
Now, we have
\begin{equation}
H(K|Y^{n})=I(K;Z^{n}|Y^{n})+H(K|Y^{n},Z^{n}).\label{boxed2}
\end{equation}
On the one hand, it holds that
\begin{align} 
 I(K;Z^{n}|Y^{n})
& \overset{(a)}{\leq }I(T^n;Z^n|Y^{n})  \nonumber \\
& \overset{(b)}{\leq } \sum_{i=1}^{n} I(Z_{i};T^n|Z^{i-1}) \nonumber \\
& \overset{(c)}{\leq}\sum_{i=1}^{n} I(T_{i};Z_{i}) \leq n C(W), \label{part1}
\end{align}
where $(a)$ follows from the fact that $I(K;Z^{n}|Y^{n})\leq I(X^{n}K;Z^{n}|Y^{n}) $ and from the Data Processing Inequality because $Y^{n}\circlearrow{X^{n}K}\circlearrow{T^n}\circlearrow{Z^{n}}$ forms a Markov chain, where we used the fact that the Data Processing inequality holds also for arbitrary random variables defined on an abstract alphabet \cite[p. 39]{dataprocessing}, $(b)$ follows from the fact that $Y^{n}\circlearrow{X^{n}K}\circlearrow{T^n}\circlearrow{Z^{n}}$ forms a Markov chain, from the fact that conditioning does not increase entropy and  from the chain rule for mutual information and $(d)$ follows because $T_{1},\dots, T_{i-1},T_{i+1},\dots, T_{n},Z^{i-1} \circlearrow{T_{i}}\circlearrow{Z_{i}}$ forms a Markov chain and because conditioning does not increase entropy.
\color{black}
On the other hand, it holds that
\begin{align}
H(K|Y^{n},Z^{n})&\overset{(a)}{\leq } H(K|L) \nonumber \\
&\overset{(b)}{\leq } 1+\log\lvert \mathcal{K} \rvert \Pr[K\neq L] \nonumber \\
&\overset{(c)}{\leq }1+\epsilon c n, \label{part2}
\end{align}
where (a) follows because $L=\Psi(Y^{n},Z^{n})$, 
(b) follows from Fano's Inequality using \eqref{errorcorrelated} and (c) follows from \eqref{cardinalitycorrelated}.

It follows from \eqref{boxed2}, \eqref{part1} and \eqref{part2} that \color{black}
$\frac{H(K|Y^n)}{n}\leq C(W)+\mu(n,\epsilon),$ with $\mu(n,\epsilon)=\frac{1}{n}+\epsilon c.$ 
We deduce that \color{black}
$
    I(U;X_{J})-I(U;Y_{J})\leq C(W)+\mu(n,\epsilon).$ \color{black}
Since the joint distribution of $X_{J}$ and $Y_{J}$ is equal to $P_{XY}$, $\frac{H(K)}{n}$ is upper-bounded by $I(U;X)$ subject to $I(U;X)-I(U;Y) \leq C(W)+\mu(n,\epsilon)$ with $P_{U|X}$ satisfying $U \circlearrow{X} \circlearrow{Y}$. As a result, for every $\epsilon,\gamma>0$ and for sufficiently large $n,$
\begin{align}
    H < \underset{ \substack{P_{U|X}\in \mc U \\{\substack{U \circlearrow{X} \circlearrow{Y}\\ I(U;X)-I(U;Y) \leq C(W)+\mu(n,\epsilon)}}}}{\sup} I(U;X)+ \gamma.
    \label{tominimize}
\end{align}
By taking the limit when $n$ tends to infinity and then the infimum over all $\epsilon,\gamma>0,$ of the right-hand side of \eqref{tominimize}, it follows that
$H\leq\underset{\alpha>0}{\inf} \underset{ \substack{P_{U|X} \in \mc U\\{\substack{U \circlearrow{X} \circlearrow{Y}\\ I(U;X)-I(U;Y) \leq C(W)+\alpha}}}}{\sup} I(U;X).$
This completes the proof of the upper-bound Theorem \ref{Theorem:Source1}.
\qed
\color{black}
\section{Conclusions} \label{sec:Conclusions}
In this paper, we investigated the problem of CR generation from bivariate sources with infinite Polish alphabets, aided by one-way communication over noisy memoryless channels. We established a single-letter lower and upper bound on the CR capacity for the specified model. \color{black}The bounds are equal except for at most countably many points where discontinuity issues might arise. However, determining when the discontinuities occur and the specific values the CR capacities take at these points remains an open problem. \color{black}

\section{Acknowledgments}
The authors acknowledge the financial support by the Federal Ministry of Education and Research
of Germany (BMBF) in the programme of “Souverän. Digital. Vernetzt.”. Joint project 6G-life, project identification number: 16KISK002. 
H. Boche and C. Deppe acknowledge the financial support
from the BMBF Quantum Programme QuaPhySI under Grant
16KIS1598K, QUIET under Grant 16KISQ093, and the QC-
CamNetz Project under Grant 16KISQ077. They were also sup-
ported by the DFG within the project "Post Shannon Theorie
und Implementierung" under Grants BO 1734/38-1 and DE 1915/2-1.
W. Labidi, R. Ezzine and C.\ Deppe were further supported in part by the BMBF within the national initiative on Post Shannon Communication (NewCom) under Grants 16KIS1003K and 16KIS1005. C. Deppe was also supported by the German Research Foundation (DFG) within the project DE1915/2-1. M. Wiese was further supported by the Bavarian Ministry of
Economic Affairs, Regional Development and Energy as part of the project 6G
Future Lab Bavaria.
\appendices

\bibliographystyle{IEEEtran}
\bibliography{definitions,references}

\begin{thebibliography}{10}
\providecommand{\url}[1]{#1}
\csname url@samestyle\endcsname
\providecommand{\newblock}{\relax}
\providecommand{\bibinfo}[2]{#2}
\providecommand{\BIBentrySTDinterwordspacing}{\spaceskip=0pt\relax}
\providecommand{\BIBentryALTinterwordstretchfactor}{4}
\providecommand{\BIBentryALTinterwordspacing}{\spaceskip=\fontdimen2\font plus
\BIBentryALTinterwordstretchfactor\fontdimen3\font minus
  \fontdimen4\font\relax}
\providecommand{\BIBforeignlanguage}[2]{{%
\expandafter\ifx\csname l@#1\endcsname\relax
\typeout{** WARNING: IEEEtran.bst: No hyphenation pattern has been}%
\typeout{** loaded for the language `#1'. Using the pattern for}%
\typeout{** the default language instead.}%
\else
\language=\csname l@#1\endcsname
\fi
#2}}
\providecommand{\BIBdecl}{\relax}
\BIBdecl

\bibitem{Sudan}
M.~Sudan, H.~Tyagi, and S.~Watanabe, ``Communication for generating
  correlation: A unifying survey,'' \emph{IEEE Transactions on Information
  Theory}, vol.~66, no.~1, pp. 5--37, 2020.

\bibitem{part2}
R.~{Ahlswede} and I.~{Csiszar}, ``Common randomness in information theory and
  cryptography. \uppercase{II. CR} capacity,'' \emph{IEEE Transactions on
  Information Theory}, vol.~44, no.~1, pp. 225--240, 1998.

\bibitem{Idchannels}
R.~{Ahlswede} and G.~{Dueck}, ``Identification via channels,'' \emph{IEEE
  Transactions on Information Theory}, vol.~35, no.~1, pp. 15--29, 1989.

\bibitem{Shannon}
C.~E. Shannon, ``A mathematical theory of communication,'' \emph{Bell System
  Technical Journal}, vol.~27, pp. 379--423, 623--656, July, October 1948.

\bibitem{trafo}
R.~Ahlswede, ``General theory of information transfer: Updated,''
  \emph{Discrete Applied Mathematics}, vol. 156, pp. 1348--1388, 05 2008.

\bibitem{ahlswede2021}
------, \emph{Watermarking Identification Codes with Related Topics on Common
  Randomness}.\hskip 1em plus 0.5em minus 0.4em\relax Cham: Springer
  International Publishing, 2021, pp. 271--325.

\bibitem{MOULINwatermarking}
P.~Moulin, ``The role of information theory in watermarking and its application
  to image watermarking,'' \emph{Signal Processing}, vol.~81, no.~6, pp. 1121
  -- 1139, 2001, special section on Information theoretic aspects of digital
  watermarking.

\bibitem{AhlswedeWatermarking}
R.~Ahlswede and N.~Cai, \emph{Watermarking Identification Codes with Related
  Topics on Common Randomness}.\hskip 1em plus 0.5em minus 0.4em\relax Berlin,
  Heidelberg: Springer Berlin Heidelberg, 2006, pp. 107--153.

\bibitem{SteinbergWatermarking}
Y.~{Steinberg} and N.~{Merhav}, ``Identification in the presence of side
  information with application to watermarking,'' \emph{IEEE Transactions on
  Information Theory}, vol.~47, no.~4, pp. 1410--1422, 2001.

\bibitem{industry4.0}
Y.~Lu, ``Industry 4.0: A survey on technologies, applications and open research
  issues,'' \emph{Journal of Industrial Information Integration}, vol.~6, pp. 1
  -- 10, 2017.

\bibitem{6Gcomm}
G.~Fettweis and H.~Boche, ``6{G}: The personal tactile internet—and open
  questions for information theory,'' \emph{IEEE BITS the Information Theory
  Magazine}, vol.~1, no.~1, pp. 71--82, 2021.

\bibitem{part1}
R.~{Ahlswede} and I.~{Csiszar}, ``Common randomness in information theory and
  cryptography. \uppercase{I}. \uppercase{S}ecret sharing,'' \emph{IEEE
  Transactions on Information Theory}, vol.~39, no.~4, pp. 1121--1132, 1993.

\bibitem{maurer}
U.~M. {Maurer}, ``Secret key agreement by public discussion from common
  information,'' \emph{IEEE Transactions on Information Theory}, vol.~39,
  no.~3, pp. 733--742, 1993.

\bibitem{semanticsecurity}
M.~Wiese and H.~Boche, ``Semantic security via seeded modular coding schemes
  and ramanujan graphs,'' \emph{IEEE Transactions on Information Theory},
  vol.~67, no.~1, pp. 52--80, 2021.

\bibitem{boche2022semantic}
H.~Boche, M.~Cai, C.~Deppe, R.~Ferrara, and M.~Wiese, ``Semantic security for
  quantum wiretap channels,'' \emph{Journal of Mathematical Physics}, vol.~63,
  no.~9, 2022.

\bibitem{researchgroup1}
F.~{Fitzek} and H.~{Boche}, ``\BIBforeignlanguage{en}{Research landscape –
  6{G} networks research in europe: 6{G}-life: Digital transformation and
  sovereignty of future communication networks},''
  \emph{\BIBforeignlanguage{en}{IEEE Network}}, vol.~35, no.~6, pp. 4--5, Nov
  2021.

\bibitem{researchgroup2}
F.~{Fitzek et. al.}, ``6{G} activities in germany,'' \emph{IEEE Future
  Networks},
  https://futurenetworks.ieee.org/tech-focus/december-2022/6g-activities-in-germany.

\bibitem{6Gpostshannon}
J.~A. Cabrera, H.~Boche, C.~Deppe, R.~F. Schaefer, C.~Scheunert, and F.~H.~P.
  Fitzek, \emph{6{G} and the Post-Shannon Theory}.\hskip 1em plus 0.5em minus
  0.4em\relax John Wiley \& Sons, Ltd, 2021, ch.~16, pp. 271--294.

\bibitem{globecom}
R.~{Ezzine}, W.~{Labidi}, H.~{Boche}, and C.~{Deppe}, ``Common randomness
  generation and identification over gaussian channels,'' in \emph{GLOBECOM
  2020 - 2020 IEEE Global Communications Conference (GLOBECOM)}, 2020, pp.
  1--6.

\bibitem{SISOfasingCR}
R.~Ezzine, M.~Wiese, C.~Deppe, and H.~Boche, ``Common randomness generation
  over slow fading channels,'' in \emph{2021 IEEE International Symposium on
  Information Theory (ISIT)}, 2021, pp. 1925--1930.

\bibitem{MIMOfadingCR}
------, ``Outage common randomness capacity characterization of
  multiple-antenna slow fading channels,'' in \emph{2021 IEEE Information
  Theory Workshop (ITW)}, 2021, pp. 1--6.

\bibitem{UCR}
------, ``A general formula for uniform common randomness capacity,'' in
  \emph{2022 IEEE Information Theory Workshop (ITW)}, 2022, pp. 762--767.

\bibitem{CRgaussiansources}
W.~Labidi, R.~Ezzine, C.~Deppe, and H.~Boche, ``Common randomness generation
  from gaussian sources,'' in \emph{2022 IEEE International Symposium on
  Information Theory (ISIT)}, 2022, pp. 1548--1553.

\bibitem{CRcountable}
W.~Labidi, R.~Ezzine, C.~Deppe, M.~Wiese, and H.~Boche, ``Common randomness
  generation from sources with countable alphabet,'' in \emph{ICC 2023 - IEEE
  International Conference on Communications}, 2023, pp. 2425--2430.

\bibitem{Mitran}
P.~Mitran, ``On a markov lemma and typical sequences for polish alphabets,''
  \emph{IEEE Transactions on Information Theory}, vol.~61, no.~10, pp.
  5342--5356, 2015.

\bibitem{klenke2020probability}
A.~Klenke, \emph{Probability Theory: A Comprehensive Course}, ser.
  Universitext.\hskip 1em plus 0.5em minus 0.4em\relax Springer International
  Publishing, 2020.

\bibitem{dembo2009large}
A.~Dembo, \emph{Large deviations techniques and applications}.\hskip 1em plus
  0.5em minus 0.4em\relax Springer, 2009.

\bibitem{mitran_polish}
P.~{Mitran}, ``{Typical Sequences for Polish Alphabets},'' \emph{arXiv
  e-prints}, p. arXiv:1005.2321, May 2010.

\bibitem{waldequation}
S.~S. David~Meintrup, \emph{Stochastik. Theorie und Anwendungen}.\hskip 1em
  plus 0.5em minus 0.4em\relax Berlin,Heidelberg: Springer, 2005, p. 287.

\bibitem{CRMIMOslowfading}
R.~Ezzine, M.~Wiese, C.~Deppe, and H.~Boche, ``Message transmission and common
  randomness generation over \uppercase{MIMO} slow fading channels with
  arbitrary channel state distribution,'' \emph{IEEE Transactions on
  Information Theory}, vol.~70, no.~1, 2024.

\bibitem{epsilonUCRcapacity}
------, ``A lower and upper bound on the epsilon-uniform common randomness
  capacity,'' in \emph{2023 IEEE International Symposium on Information Theory
  (ISIT)}, 2023, pp. 240--245.

\bibitem{dataprocessing}
S.~Ihara, \emph{Information Theory for Continuous Systems}, 1993, ch.~1.

\end{thebibliography}

\IEEEtriggeratref{4}



\end{document}